\documentclass[conference]{IEEEtran}
\usepackage[utf8]{inputenc}
\usepackage[english]{babel}
\usepackage{amsthm}
\usepackage{amsfonts,amssymb,amsthm}
\usepackage{bbm}
\usepackage{mathtools}
\usepackage{physics}
\usepackage{amsfonts}
\usepackage[ruled,vlined,linesnumbered]{algorithm2e}
\usepackage{enumerate}
\usepackage{mathtools}
\usepackage{stackengine}
\usepackage{todonotes}

\usepackage{mathrsfs}
\usepackage{balance}

\SetKwInput{KwInput}{Input}
\SetKwInput{KwOutput}{Output}
\let\oldKwInput\KwInput
\renewcommand{\KwInput}[1]{%
  \makebox[\widthof{\KwOutput{}}][l]{\oldKwInput{}#1}%
}

\newtheorem{theorem}{Theorem}

\newtheorem{lemma}{Lemma}

\theoremstyle{definition}
\newtheorem{definition}{Definition}

\newtheorem{lemma_appendix}{Lemma A\ignorespaces}

\makeatletter
\@addtoreset{proofpart}{theorem}
\makeatother

\newenvironment{proofsketch}{%
  \proof}{\endproof}

\usepackage{color}
\usepackage{listings}

\usepackage{tikz,pgfplots}
\usetikzlibrary{positioning,fit}

\usetikzlibrary{calc}
\usetikzlibrary{arrows}
\usetikzlibrary{decorations.pathreplacing,calligraphy}

\usepackage{lipsum}
\usepackage{enumitem}

\title{On Noisy Duplication Channels \\with Markov Sources}
\author{%
  \IEEEauthorblockN{Brendon McBain, James Saunderson, and Emanuele Viterbo}
  \IEEEauthorblockA{Department of Electrical \& Computer Systems Engineering, Monash University, Clayton, Australia}
}

\begin{document}

\maketitle

\begin{abstract}
    Channels with noisy duplications have recently been used to model the nanopore sequencer. This paper extends some foundational information-theoretic results to this new scenario. We prove the asymptotic equipartition property (AEP) for noisy duplication processes based on ergodic Markov processes. A consequence is that the noisy duplication channel is information stable for ergodic Markov sources, and therefore the channel capacity constrained to Markov sources is the Markov-constrained Shannon capacity. We use the AEP to estimate lower bounds on the capacity of the binary symmetric channel with Bernoulli and geometric duplications using Monte Carlo simulations. In addition, we relate the AEP for noisy duplication processes to the AEP for hidden semi-Markov processes.
\end{abstract}

\section{Introduction}
Recently, the nanopore sequencer in DNA storage systems has motivated the study of channels with noisy duplications \cite{BITSDNA}. In this application, a sequence of nucleotide molecules $\{\mathsf{A},\mathsf{T}, \mathsf{C}, \mathsf{G}\}$ pass through a microscopic pore that outputs a current response dependent on the nucleotides inside the pore. The mechanism that feeds the nucleotides through the pore varies in speed such that the sampled current signal contains a random number of samples per 
nucleotide. In the presence of measurement noise, this imperfect mechanism results in a channel with noisy duplications. This paper is an information-theoretic analysis of generic {\em noisy duplication channels}, not specifically related to the model of the nanopore sequencer.

Numerical capacity bounds of noisy duplication channels were computed in \cite{McBain2022,McBain2024} using Monte Carlo simulations to characterize the theoretical performance of the nanopore sequencer. However, the results were based on the assumption that the asymptotic equipartition property (AEP) holds for the channel processes of noisy duplication channels. The key contribution of this paper is proving that the AEP holds for discrete-time ergodic Markov sources on a finite state space. A consequence of our results is that the Markov-constrained channel capacity is the Markov-constrained Shannon capacity $C_{\mathsf{Markov}}$, which is a lower bound on the channel capacity $C$ for arbitrary sources. We compute numerical lower bounds on $C_{\mathsf{Markov}}$ for the binary symmetric channel (BSC) with Bernoulli and geometric duplications by choosing a $\mathsf{Ber}(1/2)$ source. These extend the numerical capacity bounds for the sticky channels in \cite{Mitzenmacher2007} that correspond to the noiseless cases. In addition, we relate the AEP for noisy duplication processes with the AEP for hidden semi-Markov processes (HSMPs), which has largely only been studied in the special case of semi-Markov processes (SMPs) \cite{Girardin2004, Girardin2006}.

\subsection{Noisy duplication channel}
The noisy duplication channel for the special case of Markov sources is conveniently described using SMPs. For a sequence of channel inputs $\{X_{i}\}$, the sequence of channel states is the discrete-time homogeneous Markov process (MP) $\{S_{\ell}\}$ on a finite state space $\Omega$. Since there is a one-to-one correspondence between the channel inputs and the channel states, we will conveniently consider the latter as being the input to the channel. This Markov source is {\em ergodic} if its Markov transition probability matrix $P$ is {\em irreducible}. 
For an arbitrary initial state $S_0$, their relationship is visualised as
\begin{align*}
    S_1 \overset{X_1}{\longrightarrow} S_2 \overset{X_2}{\longrightarrow} \ldots \overset{X_{m}}{\longrightarrow} S_m
\end{align*}
for $m$ channel-uses, where each arrow corresponds to an input that triggers a change in channel state. Next, the channel states are duplicated according to the i.i.d. {\em state duration process} $\{K_{\ell}\}$ on a discrete support $\Lambda$, and then sequentially concatenated together to form the SMP $\{Z_t\}$ on $\Omega$. For $m$ channel-uses, we have the duplicated states
\begin{align*}
    \underbrace{Z_1,Z_2,\ldots,Z_{T_1}}_{K_1}, \underbrace{Z_{T_1+1},\ldots,Z_{T_2}}_{K_2}, \ldots, \underbrace{Z_{T_{m-1}+1},\ldots,Z_{T_m}}_{K_m}
\end{align*}
where for convenience we additionally define the {\em jump time} $T_\ell = \sum_{i \leq \ell} K_i$ for the index of the last sample in the $\ell$-th segment, forming the jump time process $\{T_{\ell}\}$. This process is {\em semi}-Markov since the Markov property only holds at the jump time between each segment of duplications. Each duplicated state in $\{Z_t\}$ undergoes the channel mapping $f: \Omega \rightarrow \mathbb{R}$, which is identical for each duplicated state. 
Then it is corrupted by memoryless noise to give the HSMP $\{Y_t\}$. 
For a channel input $S_1^m$ of $m$ channel-uses, the channel output is $Y_1^{T_m}$ and the channel joint input-output is $(S_1^m, Y_1^{T_m})$.

\subsection{Entropy rates}

For the MP $\mathbb{S}=\{S_{\ell}\}$, the entropy rate is  $H(\mathbb{S}) = \lim_{m \rightarrow \infty} \frac{1}{m}H(S_1^m)$. The HSMP $\mathbb{Y}=\{Y_{t}\}$ has entropy rate $H(\mathbb{Y}) = \lim_{t \rightarrow \infty} \frac{1}{t} H(Y_1^{t})$. When this HSMP is randomly indexed as $Y_1^{T_m}$ using i.i.d. indexing process $\mathbb{T}=\{T_{\ell}\}$, the entropy rate is $H(\mathbb{Y}^{\mathbb{T}}) = \lim_{m \rightarrow \infty} \frac{1}{m} H(Y_1^{T_m})$. Similarly, we have the entropy rate of the joint process $(S_1^m, Y_1^{T_m})$ as $H(\mathbb{S}, \mathbb{Y}^{\mathbb{T}}) = \lim_{m \rightarrow \infty} \frac{1}{m} H(S_1^m, Y_1^{T_m})$. The existence of these two randomly-indexed entropy rates is proven in Lemma 1. Finally, let the (mutual) information rate of the noisy duplication channel with a Markov source be $I(\mathbb{S}; \mathbb{Y}^{\mathbb{T}}) = H(\mathbb{S}) + H(\mathbb{Y}^{\mathbb{T}}) - H(\mathbb{S},\mathbb{Y}^{\mathbb{T}})$.

\begin{lemma}
    The entropy rates $H(\mathbb{Y}^{\mathbb{T}})$ and $H(\mathbb{S},\mathbb{Y}^{\mathbb{T}})$ exist.
    \label{ent_limit_exists}
\end{lemma}
\begin{proof}
    Let $H_m = \frac{1}{m}H(Y_{1}^{T_{m}})$. Observe that
    \begin{align}
        m H_m &= H(Y_{1}^{T_{m}})\\
        &\leq H(Y_{1}^{T_{\ell}},Y_{T_{\ell + 1}}^{T_{m}}) ~~\text{for any}~\ell<m\\
        &= H(Y_{1}^{T_{\ell}}) + H(Y_{T_{\ell} + 1}^{T_{m}} | Y_{1}^{T_{\ell}}) \\
        &\leq H(Y_{1}^{T_{\ell}}) + H(Y_{T_{\ell} + 1}^{T_{m}})\\
        &= \ell H_{\ell} + (m-\ell) H_{m-\ell}
    \end{align}
    then $\{m H_m\}$ is sub-additive.
    By Fekete's lemma \cite[~Lemma 4A.2]{Gallager1968}, we have that $\lim_{m\rightarrow\infty}H_m = \liminf_{m \rightarrow \infty} H_m$ exists. Each of the above steps similarly applies when setting $H_m = \frac{1}{m}H(S_1^m, Y_{1}^{T_{m}})$.
\end{proof}

\section{Asymptotic equipartition property}

\subsection{Noisy duplication processes}
In this section, we consider the AEP for the output entropy rate $H(\mathbb{Y}^{\mathbb{T}})$ and the joint input-output entropy rate $H(\mathbb{S}, \mathbb{Y}^{\mathbb{T}})$. These results will be proven using almost identical techniques based on genie-aided markers. Since we cannot directly use ergodicity to prove the AEP through the Shannon-McMillan-Breiman theorem (SMB) \cite{Algoet1988}, we use the side-information of genie-aided markers to force ergodicity in the noisy duplication processes with respect to the sub-sequences between the markers in order to apply SMB and indirectly prove the AEP.

\begin{theorem}[Output AEP]
If $\{S_{\ell}\}$ is ergodic, then
\begin{align}
    -\frac{1}{m}\log \mathbb{P}(Y_1^{T_m}) \xrightarrow[]{\mathbb{P}}  H(\mathbb{Y}^{\mathbb{T}}) .
\end{align}
\label{output_AEP}
\end{theorem}

The full proof of Theorem \ref{output_AEP} is given in Appendix \ref{appendixB}, which we now summarise as a proof sketch:

\begin{proofsketch} Let $\mathcal{M}_{m,d} = \{T_{nd} : nd \leq m \}$ be the markers up to the $m$-th symbol separated by $d$ symbols. Let $n = |\mathcal{M}_{m,d}|$ be the number of markers up to $m$, then we have $m = nd + i$ for $0 \leq i < d$. Let the sample entropy rate of $Y_1^{T_m}$ be $g_m = -\frac{1}{m}\log \mathbb{P}(Y_1^{T_m})$, and let the entropy rate be $H_{m} = \mathbb{E}[g_m] = \frac{1}{m} H(Y_1^{T_m})$ with limit $H = \lim_{m\rightarrow\infty} H_m$ (which exists from Lemma \ref{ent_limit_exists}). For a fixed $d$, let the sample entropy rate of $(Y_1^{T_m}, \mathcal{M}_{m,d})$ be $g_{m,d} = -\frac{1}{m}\log \mathbb{P}(Y_1^{T_m}, \mathcal{M}_{m,d})$, and let the entropy rate be $H_{m,d} = \mathbb{E}[g_{m,d}] = \frac{1}{m} H(Y_1^{T_m}, \mathcal{M}_{m,d})$ with limit $H_{\infty,d} = \lim_{m\rightarrow\infty} H_{m,d}$.

    \begin{itemize}
        \item \textbf{Split problem into three parts:} We need to show that $g_m \rightarrow H$ as $m\rightarrow\infty$. By the triangle inequality, we need only show the chain of convergences $g_m \rightarrow g_{m,d} \rightarrow H_{\infty,d} \rightarrow H$ for some $d=d(m)$ as $m\rightarrow\infty$.
        \item \textbf{Convergence of entropy rates with markers:} Adding markers can only increase the entropy rate such that $H_{m,d} - H_m \geq 0$. This increase in entropy rate is bounded as $H_{m,d} - H_m \leq \frac{2}{d}H(T_d)$ by Lemma A\ref{ent_diff_ineq} and converges to zero as the marker distance $d$ increases by Lemma A\ref{ent_jump_time}.
        \item \textbf{Convergence of sample entropy rates with markers:} Observe that $g_{m,d}$ is decreasing towards $g_m$ for increasing $d$ since the side-information of markers reduces the number of terms in the marginalisation of $T_1^m$. Then $g_{m,d} - g_{m} \geq 0$ and Markov's inequality says they converge since $\mathbb{E}[g_{m,d} - g_{m}] = H_{m,d} - H_m$.
        \item \textbf{Shannon-McMillan-Breiman theorem with markers:} Form a hidden MP (HMP) $\{W_n\}$ with $W_n = Y_{T_{(n-1)d}+1}^{T_{nd}}$ based on the MP $\{V_n\}$ with $V_n = S_{(n-1)d + 1}^{n d}$ in between the markers. The HMP is ergodic since the MP is ergodic~\cite{Leroux1992}. Then $g_{m,d} \rightarrow H_{\infty,d}$ as $m\rightarrow\infty$ by SMB.
    \end{itemize}\end{proofsketch}

\begin{theorem}[Joint AEP]
If $\{S_{\ell}\}$ is ergodic, then

\begin{align}
    -\frac{1}{m}\log \mathbb{P}(S_1^m, Y_1^{T_m}) \xrightarrow[]{\mathbb{P}}  H(\mathbb{S}, \mathbb{Y}^{\mathbb{T}}) .
\end{align}
\label{joint_AEP}
\end{theorem}

\begin{proof}This proof only differs from Theorem \ref{output_AEP} in that instead of showing ergodicity of the HMP $\{W_n\}$ we need only show ergodicity of the process $\{U_n\}$ for $U_n = (V_n, W_n)$. This is also an ergodic HMP with respect to ergodic MP $\{V_n\}$. It is actually a MP, however it is more convenient to invoke the ergodicity of HMPs \cite{Leroux1992}. \end{proof}

There are many consequences of these AEP results. In Section \ref{markov_cap}, we will use them to derive achievable lower bounds of channel capacity and demonstrate how they can be estimated in practice. Before that, let us explore another AEP with an interesting connection to the output AEP.

\subsection{Hidden semi-Markov processes}
For noisy duplication channels, we were interested in the case where there are $m$ inputs and $T_m$ duplicated samples at the output. Alternatively, we can choose the number of outputs to be $t$ samples, and therefore the number of inputs is a random variable $M(t)$ much smaller than $t$. We consider the AEP for discrete-time HSMPs over a finite state-space, extending the AEP for discrete-time SMPs over a finite state-space \cite{Girardin2004} (and was later extended to the Borel state-space \cite{Girardin2006}, which is not considered in this paper). The central idea behind proving our AEP will be based on embedding the SMP in a MP, for which we can leverage existing results.

\begin{definition}[Embedded SMP \cite{Johnson2014}]
The embedded SMP is a MP with transition probabilities
\begin{align}
    \begin{split}
        P(s, k| s', k') &= \begin{cases} 
      P(s|s') \mathbb{P}(K=k) &  k=1\\
      \mathbb{P}(K > k|K > k-1) & k = k' + 1, s = s'\\
      0 & \text{otherwise}
   \end{cases}
   \end{split}
\end{align}
for all $k',k \in \Lambda$ and for all $s', s \in \Omega$. 
\end{definition}

 An important property of the embedded duplication process is that it preserves ergodicity. Observe that the embedded duration for each channel input form a sequence of states with one possible transition, the ``extending'' transition, until it reaches the ``terminating'' transition. Therefore, the embedding MP is irreducible and thus obeys the ergodic theorem. This embedding of the SMP will be used to prove the AEP for HSMPs in Theorem \ref{hsmp_aep}. Further, this HSMP AEP will be related to the output AEP in Theorem \ref{output_AEP} through its entropy rate using Lemma \ref{lem:rand_indexed_ent_rate}.

\begin{lemma}[Randomly indexed entropy rate]
Let $H(\mathbb{Y}^{\mathbb{T}}) = \lim_{m \rightarrow \infty}\frac{1}{m}H(Y_1^{T_m})$ be the entropy rate for fixed-length inputs and variable-length outputs, and let $H(\mathbb{Y}) = \lim_{t \rightarrow \infty}\frac{1}{t}H(Y_1^{t})$ be the entropy rate for variable-length inputs and fixed-length outputs. Then $H(\mathbb{Y}^{\mathbb{T}}) = \mathbb{E}[K] H(\mathbb{Y})$.
\label{lem:rand_indexed_ent_rate}
\end{lemma}

\begin{proof}
    Let $\tilde{Z}_1^t$ be an embedding of the randomly indexed MP $S_1^{M(t)}$.  Observe that the number of random inputs for $t$ output samples is $M(t) = \sum_{t} \mathbf{1}_{\mathcal{D}}[\tilde{Z}_t]$ where $\mathcal{D} = \{(s,k) \in \mathcal{S}\times \Lambda : k = 1\}$ is the set of embedding states that begin a new segment. With this, we will show that the entropy rate $H(M(t))/t$ goes to zero as $t \rightarrow \infty$.

    By Hoeffding's inequality for MPs \cite{Fan2021}, we have the concentration inequality $\mathbb{P}(|M(t) - \mathbb{E}[M(t)]| > t) \leq 2 \exp(-C t)$ for some constant $C > 0$ that depends on the convergence speed of the MP to its stationary distribution. Let $\mathcal{A}_t=\{|M(t) - \mathbb{E}[M(t)]| > t \}$, then we can partition the entropy of $M(t)$ and bound it as
    \begin{align}
    \begin{split}
        \frac{1}{t} H(M(t)) &= \frac{1}{t}H(M(t)|\mathcal{A}_t)\mathbb{P}(\mathcal{A}_t) \\
        &\quad\quad + \frac{1}{t}H(M(t)|\mathcal{A}^c_t)(1-\mathbb{P}(\mathcal{A}_t))
    \end{split}\\
    &\leq \frac{1}{t}H(M(t)) 2 e^{-C t} + \frac{1}{t}\log(C t)(1-2 e^{-C t})
\end{align}
 which uses the conditioning property, Hoeffding's inequality, and that $H(M(t) | \mathcal{A}^c_t) \leq \log (2 C t)$ since $M(t)$ conditioned on $\mathcal{A}^c_t$ is supported on the interval $[\mathbb{E}[M(t)] - t, \mathbb{E}[M(t)] + t]$.

 Observe that $|t - T_{M(t)}| \leq a$ for constant $a = \max \Lambda$ and for all $t$. Then $H(Y_1^t)/t$ is arbitrarily close to $H(Y_1^{T_{M(t)}})/t$ for sufficiently large $t$, and
 \begin{align}
     H(\mathbb{Y}) &= \lim_{t\rightarrow\infty} \frac{1}{t} H(Y_1^{T_{M(t)}})\\
     &= \lim_{t\rightarrow\infty} \frac{1}{t}[H(Y_1^{T_{M(t)}}|M(t)) + H(M(t))] \\
     &= \lim_{t\rightarrow\infty} \frac{1}{t}H(Y_1^{T_{M(t)}}|M(t))\\
     &= \lim_{t\rightarrow\infty} \sum_{m}\mathbb{P}(M(t) = m) \left[\frac{H(Y_1^{T_{m}})}{m \mathbb{E}[K]}\right] \label{eq:ri_ent_rate_limit}
 \end{align}
where (\ref{eq:ri_ent_rate_limit}) uses a sandwich bound on the conditional entropy $m H_m = H(Y_1^{T_m}) = H(Y_1^{T_{M(t)}}|M(t)=m)$, given by
\begin{align}
    \left(\frac{m}{T_{M(t)} + a}\right) H_m &\leq \frac{1}{t} H(Y_1^{T_m}) \leq \left(\frac{m}{T_{M(t)} - a}\right) H_m
\end{align}
which is derived from $|t - T_{M(t)}| \leq a$. By the (randomly-indexed) law of large numbers $T_{M(t)} / M(t) \rightarrow \mathbb{E}[K]$ as $t\rightarrow\infty$, we can squeeze $H(Y_1^{T_m})/t$ and show it has the limit $H(\mathbb{Y}^{\mathbb{T}})/\mathbb{E}[K]$. Since $\mathbb{E}[M(t)]$ is increasing in $t$, the left-hand tail of $\mathbb{P}(M(t) = m)$ is arbitrarily small up to any given $m$ for sufficiently large $t$, showing convergence of (\ref{eq:ri_ent_rate_limit}).\end{proof}

\begin{theorem}[HSMP AEP] \label{hsmp_aep}
If $\{Z_t\}$ is ergodic, then
\begin{align}
    -\frac{1}{t}\log \mathbb{P}(Y_1^{t}) \xrightarrow[]{a.s.} H(\mathbb{Y}) = \frac{H(\mathbb{Y}^{\mathbb{T}})}{\mathbb{E}[K]}.
\end{align}
\end{theorem}

\begin{proof}
Embed the SMP $Z_1^{t}$ on $\Omega$ into the MP $\tilde{Z}_1^t$ on $\Omega \times \Lambda$, such that the corresponding embedding output process $\tilde{Y}_1^t$ is a HMP. Since the embedding MP is ergodic, the AEP follows from the ergodicity of HMPs \cite{Leroux1992} and SMB, combined with Lemma \ref{lem:rand_indexed_ent_rate}.
\end{proof}

This theorem combined with the AEP for the output of the noisy duplication channel shows equivalence between the two AEPs. An interesting consequence of this observation is that it implies the true block length $M(t)$, which is unknown from the $t$ observed samples, dominates the marginalisation of $T_1^{M(t)}$ in the entropy rate. 
It is not clear if there are any practical implications from this, but it is nonetheless an interesting observation. It does, however, allow us to use the classical forward algorithm to estimate $H(\mathbb{Y}^{\mathbb{T}})$ with the embedding using techniques for finite-state channels \cite{Pfister2001,Arnold2001}.

\section{Markov-constrained channel capacity} \label{markov_cap}

An open problem of the noisy duplication channel is its channel capacity. If the channel were ergodic, the channel capacity  would be the Shannon capacity formula. However, this must be proven directly since we do not have ergodicity as a given. In the special case of ergodic Markov sources, this result turns out to be a corollary of the AEPs in Theorem \ref{output_AEP} and Theorem \ref{joint_AEP} since they imply information stability. 

\begin{theorem}
    If the Markov source $\{S_{\ell}\}$ with Markov transition matrix $P$ is ergodic, then the Markov-constrained channel capacity is the Markov-constrained Shannon capacity
\begin{align}
    C_{\text{Markov}} = \sup_{P \in \mathcal{P}} I(\mathbb{S}; \mathbb{Y}^{\mathbb{T}})
\end{align}
where $\mathcal{P}$ is the set of all ergodic Markov transition matrices.
\end{theorem}
\begin{proof}
    Consider the information density $$i_m = \frac{1}{m} \log \mathbb{P}(S_1^m, Y_1^{T_m}) - \frac{1}{m} \log \mathbb{P}(Y_1^{T_m}) - \frac{1}{m} \log \mathbb{P}(S_1^m), $$ whose expectation is the finite-letter mutual information rate $\frac{1}{m} I(S_1^m; Y_1^{T_m})$. Observe that the first and second terms converge in probability to $H(\mathbb{S}, \mathbb{Y}^{\mathbb{T}})$ and $H(\mathbb{Y}^{\mathbb{T}})$, respectively, due to Theorem \ref{output_AEP} and Theorem \ref{joint_AEP},  and that the third term converges to $H(\mathbb{S})$ due to SMB. Consequently, the noisy duplication channel with an ergodic Markov source is {\em information stable}. 
 This is a special case of the general capacity formula in \cite{Verdu1994}, and implies that the Markov-constrained Shannon capacity is achievable using the capacity-achieving Markov source.
\end{proof}

 For channel capacity $C$ with an arbitrary source, we have $C_{\text{Markov}} \leq C$ since the optimal source may not be Markov. Even in the restricted case of Markov sources, finding the Markov source that achieves $C_{\text{Markov}}$ is an open problem. However, there exist heuristic algorithms \cite{McBain2022,McBain2024} based on the generalised Blahut-Arimoto algorithm (GBAA) \cite{Kavcic2001,Vontobel2008} for finite-state channels that can significantly improve upon the benchmark Markov
source with independent and uniformly distributed inputs (i.e., the maximum-entropy Markov source). 

In the following examples, we consider some simple capacity bounds of a noisy duplication channel that concatenates the BSC with the sticky channel in \cite{Mitzenmacher2007} for Bernoulli and geometric duplications.

\subsection{BSC with Bernoulli duplications}
The simplest non-trivial noisy duplication channel is a BSC with Bernoulli duplications. In particular, we consider a BSC with crossover probability $p$ and durations $K_{\ell} \sim \mathsf{Ber}(p_d)$ on $\Lambda = \{1,2\}$ with duplication probability $p_d$. Each source bit is duplicated with probability $p_d$, and then each bit goes through the BSC that corrupts it with probability $p$.

In Fig. \ref{fig:bsc_d}, we use the Monte Carlo technique from \cite{McBain2024} with a block length of $m=10^6$ to compute the achievable information rate $I_{\text{BSCD,Ber}}(p,p_d)$ of this channel in the case of a $\mathsf{Ber}(1/2)$ source. When $p_d=0$, the rate is equal to the capacity of the BSC, $C_{\text{BSC}}(p)$. When $p_d=1$, the rate is equal to the capacity of the BSC with $2$ looks, $C_{\text{BSC}^2}(p)$.
The red curve is when $p=0.1$, which starts at $C_{\text{BSC}}(0.1) = 0.5310$ and ends at $C_{\text{BSC}^2}(0.1) = 0.7421$. The blue curve is when $p=0.01$, which starts at $C_{\text{BSC}}(0.01) = 0.9192$ and ends at $C_{\text{BSC}^2}(0.01) = 0.9787$. When $p=0$ the BSC with Bernoulli duplications simplifies into a sticky channel with Bernoulli duplications with rate $I_{\text{SC,Ber}}(p_d)$ for a $\mathsf{Ber}(1/2)$ source and channel capacity $C_{\text{SC,Ber}}(p_d)$ for the capacity-achieving source, which can be computed numerically \cite{Mitzenmacher2007}. The capacity of the sticky channel with Bernoulli duplications is an upper bound on the capacity of the BSC with Bernoulli duplications. 

\subsection{BSC with geometric duplications}
We now consider a BSC with geometric duplications, which is a slightly more complex channel compared to the BSC with Bernoulli duplications. In particular, we consider a BSC with crossover probability $p$ and durations $K_{\ell} \sim \mathsf{Geom}(p_d)$ on $\Lambda = \{1,2, \ldots\}$ with duplication probability $p_d$.
Each source bit is repeatedly duplicated with probability $p_d$ until the first non-duplication with probability $1-p_d$, and then each bit goes through the BSC that corrupts it with probability $p$.

In Fig. \ref{fig:bsc_geom}, we once again use the Monte Carlo technique from \cite{McBain2024} with a block length of $m=10^6$ to compute the achievable information rate $I_{\text{BSCD,geom}}(p,p_d)$ of this channel in the case of a $\mathsf{Ber}(1/2)$ source. As noted in \cite{Mitzenmacher2007}, accurately computing rates with geometric duplications and a high $p_d$ is challenging, and therefore we stop at $p_d = 0.6$ (just over $2$ duplications on average) as they did. In addition, the geometric distribution is truncated after $15$ samples. Analogously to the previous example, when $p=0$ the BSC with geometric duplications simplifies into a sticky channel with geometric duplications with rate $I_{\text{SC,geom}}(p_d)$ for a $\mathsf{Ber}(1/2)$ source and channel capacity $C_{\text{SC,geom}}(p_d)$ for the capacity-achieving source, which can be computed numerically \cite{Mitzenmacher2007}. The capacity of the sticky channel with geometric duplications is an upper bound on the capacity of the BSC with geometric duplications.

\begin{figure}
    \centering

\begin{tikzpicture}[scale=.49]

\begin{axis}[%
width=6.028in,
height=4.754in,
at={(1.011in,0.642in)},
scale only axis,
xmin=0,
xmax=1,
xlabel style={font=\color{white!15!black}},
xlabel={$p_d$},
title style={font={\bfseries\Large}},
title={BSC with Bernoulli duplications},
ymin=0,
ymax=1,
ylabel style={font=\color{white!15!black}},
ylabel={Information rate [bits/symbol]},
axis background/.style={fill=white},
legend style={at={(0.97,0.03)}, anchor=south east, legend cell align=left, align=left, draw=white!15!black},
xmajorgrids,
ymajorgrids,
grid style = {
    dash pattern = on 0.025mm off 0.95mm on 0.025mm off 0mm, 
    line cap = round,
    gray,
    line width = 0.5pt
}
]
\addplot [color=black, dashdotted, line width=1.5pt]
  table[row sep=crcr]{%
0	0.999999999999999\\
0	0.999999999999999\\
0.01	0.941890816260903\\
0.02	0.903529204314827\\
0.03	0.872471827999638\\
0.04	0.846079220805141\\
0.05	0.823101018372947\\
0.06	0.802803008110957\\
0.07	0.784701715398122\\
0.08	0.768454823942257\\
0.09	0.75380745550931\\
0.1	0.740562685327654\\
0.11	0.728563998867609\\
0.12	0.717684198555393\\
0.13	0.707818046448059\\
0.14	0.698877193074358\\
0.15	0.690786567462707\\
0.16	0.683481730754314\\
0.17	0.676906871311268\\
0.18	0.671013210040627\\
0.19	0.665757619870457\\
0.2	0.661101226459302\\
0.21	0.65700779269788\\
0.22	0.653442811202527\\
0.23	0.650372313479935\\
0.24	0.647759279141918\\
0.25	0.645573186948353\\
0.26	0.643785825459013\\
0.27	0.642369477964945\\
0.28	0.641297395615439\\
0.29	0.640544408066522\\
0.3	0.640087558140451\\
0.31	0.63990673229976\\
0.32	0.639985138489138\\
0.33	0.640309481917889\\
0.34	0.640869785055021\\
0.35	0.641658899984763\\
0.36	0.642671828581371\\
0.37	0.643905001208042\\
0.38	0.645355652598419\\
0.39	0.647021359848146\\
0.4	0.648899713085444\\
0.41	0.650988087834434\\
0.42	0.653283723342928\\
0.43	0.655783967090158\\
0.44	0.658488238206166\\
0.45	0.661395969107992\\
0.46	0.664504337356062\\
0.47	0.667810347538328\\
0.48	0.671311134893091\\
0.49	0.675003754492474\\
0.5	0.678885048834423\\
0.51	0.682951748872949\\
0.52	0.687200580824393\\
0.53	0.691628326329581\\
0.54	0.696231856350294\\
0.55	0.701008148345336\\
0.56	0.705954289487047\\
0.57	0.711067468005852\\
0.58	0.7163449550933\\
0.59	0.721784079843452\\
0.6	0.727382199484402\\
0.61	0.733136666790729\\
0.62	0.739044796145644\\
0.63	0.745103829277136\\
0.64	0.751310901261846\\
0.65	0.757663007009785\\
0.66	0.764156968137317\\
0.67	0.770789399913255\\
0.68	0.777556677816577\\
0.69	0.784454903156056\\
0.7	0.791479867148484\\
0.71	0.798627012810389\\
0.72	0.805891393967659\\
0.73	0.813267630614103\\
0.74	0.820749859742125\\
0.75	0.828331680618995\\
0.76	0.83600609328077\\
0.77	0.843765428753641\\
0.78	0.851601269171559\\
0.79	0.859504355520567\\
0.8	0.867464480169633\\
0.81	0.875470360606313\\
0.82	0.883509489816051\\
0.83	0.891567957446935\\
0.84	0.899630234149957\\
0.85	0.907678909109032\\
0.86	0.915694367477163\\
0.87	0.923654389815051\\
0.88	0.931533649015048\\
0.89	0.939303070515965\\
0.9	0.94692900717037\\
0.91	0.954372157825779\\
0.92	0.961586123270611\\
0.93	0.968515434306943\\
0.94	0.975092783989093\\
0.95	0.981235004202728\\
0.96	0.986836935468206\\
0.97	0.991761436464375\\
0.98	0.995821283004047\\
0.99	0.998739377520123\\
1	0.999999999999999\\
1	0.999999999999999\\
};
\addlegendentry{$C_{\text{SC,Ber}}(p_d)$}
\addplot [color=black, line width=1.5pt]
  table[row sep=crcr]{%
0	0.999999999999999\\
0	0.999999999999999\\
0.01	0.941877500251233\\
0.02	0.903475375544611\\
0.03	0.872349689974538\\
0.04	0.84586066849119\\
0.05	0.822757881150872\\
0.06	0.802307245010562\\
0.07	0.78402556371694\\
0.08	0.767570902872386\\
0.09	0.752688826835084\\
0.1	0.739182873636537\\
0.11	0.726896973998443\\
0.12	0.715704323235117\\
0.13	0.705499992242934\\
0.14	0.69619582820417\\
0.15	0.68771682174995\\
0.16	0.679998448886296\\
0.17	0.672984681452433\\
0.18	0.666626468535992\\
0.19	0.660880557429126\\
0.2	0.655708564385075\\
0.21	0.651076232462292\\
0.22	0.646952831722163\\
0.23	0.64331066928925\\
0.24	0.640124685283847\\
0.25	0.63737211665281\\
0.26	0.635032215249799\\
0.27	0.633086009674533\\
0.28	0.631516102717327\\
0.29	0.630306498006517\\
0.3	0.629442450783657\\
0.31	0.628910338748068\\
0.32	0.628697549698866\\
0.33	0.628792383316407\\
0.34	0.62918396490832\\
0.35	0.629862169328595\\
0.36	0.630817553584416\\
0.37	0.632041296891837\\
0.38	0.633525147140727\\
0.39	0.635261372891749\\
0.4	0.637242720160887\\
0.41	0.639462373356245\\
0.42	0.641913919821896\\
0.43	0.64459131751821\\
0.44	0.647488865430041\\
0.45	0.650601176345751\\
0.46	0.65392315169293\\
0.47	0.657449958152426\\
0.48	0.661177005801921\\
0.49	0.665099927564734\\
0.5	0.669214559759584\\
0.51	0.673516923563116\\
0.52	0.678003207209768\\
0.53	0.682669748763096\\
0.54	0.687513019299504\\
0.55	0.692529606349418\\
0.56	0.697716197442519\\
0.57	0.703069563602706\\
0.58	0.708586542634927\\
0.59	0.714264022039922\\
0.6	0.720098921383921\\
0.61	0.726088173938401\\
0.62	0.732228707389578\\
0.63	0.738517423398188\\
0.64	0.744951175766589\\
0.65	0.751526746941691\\
0.66	0.75824082254785\\
0.67	0.765089963602565\\
0.68	0.772070576018275\\
0.69	0.779178876934203\\
0.7	0.786410857350888\\
0.71	0.793762240454363\\
0.72	0.801228434913487\\
0.73	0.808804482308859\\
0.74	0.816484997699542\\
0.75	0.824264102148143\\
0.76	0.83213534579647\\
0.77	0.840091619802376\\
0.78	0.848125055097358\\
0.79	0.856226905485593\\
0.8	0.864387412048935\\
0.81	0.872595645115021\\
0.82	0.880839319130978\\
0.83	0.889104574597074\\
0.84	0.897375719638215\\
0.85	0.905634921689234\\
0.86	0.91386183689084\\
0.87	0.922033160809394\\
0.88	0.930122078440594\\
0.89	0.938097583216513\\
0.9	0.945923622458467\\
0.91	0.953558007644582\\
0.92	0.96095099730878\\
0.93	0.968043408699098\\
0.94	0.974764022264008\\
0.95	0.98102586647033\\
0.96	0.986720599773331\\
0.97	0.991709326924328\\
0.98	0.995805694796918\\
0.99	0.998737675752777\\
1	0.999999999999999\\
1	0.999999999999999\\
};
\addlegendentry{$I_{\text{SC,Ber}}(p_d) = I_{\text{BSCD,Ber}}(0, p_d)$}

\addplot [color=blue, line width=1.5pt]
  table[row sep=crcr]{%
0	0.919425\\
0.01	0.862165\\
0.02	0.824648\\
0.03	0.794061\\
0.04	0.767522\\
0.05	0.745523\\
0.06	0.725767\\
0.07	0.708807\\
0.08	0.694028\\
0.09	0.677828\\
0.1	0.666242\\
0.11	0.652945\\
0.12	0.643419\\
0.13	0.632729\\
0.14	0.623329\\
0.15	0.615635\\
0.16	0.609151\\
0.17	0.601876\\
0.18	0.595262\\
0.19	0.589514\\
0.2	0.585397\\
0.21	0.581482\\
0.22	0.577946\\
0.23	0.574055\\
0.24	0.572244\\
0.25	0.56895\\
0.26	0.567253\\
0.27	0.565275\\
0.28	0.56321\\
0.29	0.561983\\
0.3	0.560157\\
0.31	0.559244\\
0.32	0.560264\\
0.33	0.560466\\
0.34	0.560731\\
0.35	0.560388\\
0.36	0.562932\\
0.37	0.563906\\
0.38	0.564969\\
0.39	0.565355\\
0.4	0.568157\\
0.41	0.570508\\
0.42	0.573238\\
0.43	0.574596\\
0.44	0.576952\\
0.45	0.580718\\
0.46	0.585131\\
0.47	0.58754\\
0.48	0.5909\\
0.49	0.595514\\
0.5	0.597714\\
0.51	0.602599\\
0.52	0.606358\\
0.53	0.60987\\
0.54	0.615837\\
0.55	0.620999\\
0.56	0.624575\\
0.57	0.630031\\
0.58	0.636433\\
0.59	0.643224\\
0.6	0.647512\\
0.61	0.653002\\
0.62	0.658412\\
0.63	0.66575\\
0.64	0.671703\\
0.65	0.679487\\
0.66	0.683964\\
0.67	0.692249\\
0.68	0.698324\\
0.69	0.70652\\
0.7	0.712325\\
0.71	0.720936\\
0.72	0.727981\\
0.73	0.737815\\
0.74	0.744643\\
0.75	0.751918\\
0.76	0.759078\\
0.77	0.769629\\
0.78	0.776799\\
0.79	0.785428\\
0.8	0.795054\\
0.81	0.800787\\
0.82	0.811894\\
0.83	0.820471\\
0.84	0.830012\\
0.85	0.840997\\
0.86	0.849386\\
0.87	0.859984\\
0.88	0.869508\\
0.89	0.877965\\
0.9	0.888353\\
0.91	0.898493\\
0.92	0.908995\\
0.93	0.918068\\
0.94	0.928341\\
0.95	0.939292\\
0.96	0.948512\\
0.97	0.957416\\
0.98	0.96636\\
0.99	0.973811\\
1	0.978744\\
};
\addlegendentry{$I_{\text{BSCD,Ber}}(0.01, p_d)$}

\addplot [color=red, line width=1.5pt]
  table[row sep=crcr]{%
0	0.532509\\
0.01	0.47789\\
0.02	0.448388\\
0.03	0.421977\\
0.04	0.403042\\
0.05	0.384971\\
0.06	0.368395\\
0.07	0.358712\\
0.08	0.346456\\
0.09	0.335641\\
0.1	0.325525\\
0.11	0.317455\\
0.12	0.309907\\
0.13	0.302894\\
0.14	0.29874\\
0.15	0.29193\\
0.16	0.287674\\
0.17	0.284199\\
0.18	0.280506\\
0.19	0.276632\\
0.2	0.273746\\
0.21	0.27156\\
0.22	0.27078\\
0.23	0.268331\\
0.24	0.266947\\
0.25	0.265817\\
0.26	0.264026\\
0.27	0.262664\\
0.28	0.263681\\
0.29	0.26167\\
0.3	0.260776\\
0.31	0.261375\\
0.32	0.263305\\
0.33	0.263071\\
0.34	0.263064\\
0.35	0.264377\\
0.36	0.266146\\
0.37	0.266553\\
0.38	0.268938\\
0.39	0.269726\\
0.4	0.271506\\
0.41	0.272421\\
0.42	0.275643\\
0.43	0.27725\\
0.44	0.279652\\
0.45	0.281976\\
0.46	0.285918\\
0.47	0.287823\\
0.48	0.289336\\
0.49	0.293156\\
0.5	0.296088\\
0.51	0.299938\\
0.52	0.302765\\
0.53	0.30698\\
0.54	0.310796\\
0.55	0.314676\\
0.56	0.317313\\
0.57	0.323229\\
0.58	0.327528\\
0.59	0.333071\\
0.6	0.336018\\
0.61	0.341658\\
0.62	0.345193\\
0.63	0.351249\\
0.64	0.358005\\
0.65	0.363889\\
0.66	0.367771\\
0.67	0.374729\\
0.68	0.379249\\
0.69	0.386173\\
0.7	0.392062\\
0.71	0.399888\\
0.72	0.406101\\
0.73	0.414838\\
0.74	0.420182\\
0.75	0.427576\\
0.76	0.435437\\
0.77	0.445908\\
0.78	0.453219\\
0.79	0.460934\\
0.8	0.470494\\
0.81	0.479213\\
0.82	0.48894\\
0.83	0.49799\\
0.84	0.508353\\
0.85	0.520208\\
0.86	0.530675\\
0.87	0.542975\\
0.88	0.553706\\
0.89	0.565927\\
0.9	0.579253\\
0.91	0.593286\\
0.92	0.608318\\
0.93	0.620462\\
0.94	0.637508\\
0.95	0.653528\\
0.96	0.669356\\
0.97	0.686868\\
0.98	0.707579\\
0.99	0.725721\\
1	0.742043\\
};
\addlegendentry{$I_{\text{BSCD,Ber}}(0.1, p_d)$}
\end{axis}
\end{tikzpicture}%
    
    \caption{Monte Carlo estimates of the information rate $I_{\text{BSCD,Ber}}(p,p_d)$ of the BSC with error probability $p$, Bernoulli duplications with probability $p_d$, and a $\mathsf{Ber}(1/2)$ source. The information rate $I_{\text{SC,Ber}}(p_d)$ is the case when $p=0$, which corresponds to a sticky channel with capacity $C_{\text{SC,Ber}}(p_d)$ and is computed numerically \cite{Mitzenmacher2007}. 
    }
    \label{fig:bsc_d}
\end{figure}
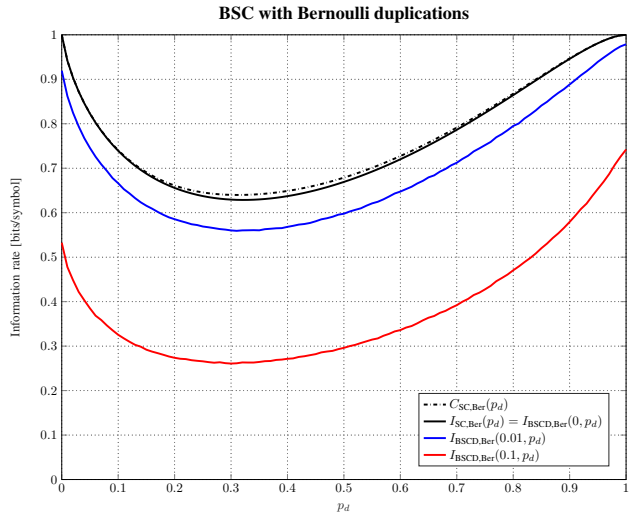

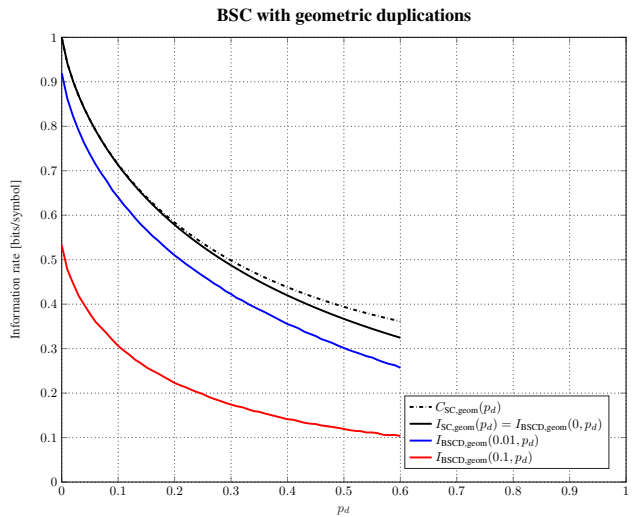
\begin{figure}
    \centering

\begin{tikzpicture}[scale=.49]

\begin{axis}[%
width=6.028in,
height=4.754in,
at={(1.011in,0.642in)},
scale only axis,
xmin=0,
xmax=1,
xlabel style={font=\color{white!15!black}},
xlabel={$p_d$},
title style={font=\bfseries\Large},
title={BSC with geometric duplications},
ymin=0,
ymax=1,
ylabel style={font=\color{white!15!black}},
ylabel={Information rate [bits/symbol] },
axis background/.style={fill=white},
legend style={at={(0.97,0.03)}, anchor=south east, legend cell align=left, align=left, draw=white!15!black},
xmajorgrids,
ymajorgrids,
grid style = {
    dash pattern = on 0.025mm off 0.95mm on 0.025mm off 0mm, 
    line cap = round,
    gray,
    line width = 0.5pt
}
]
\addplot [color=black, dashdotted, line width=1.5pt]
  table[row sep=crcr]{%
0	0.999977981012336\\
0.01	0.941314466120063\\
0.02	0.901657816739322\\
0.03	0.868763457742876\\
0.04	0.840087870397167\\
0.05	0.814442957612982\\
0.06	0.791139513457703\\
0.07	0.76972860999884\\
0.08	0.749895294803979\\
0.09	0.731407218577373\\
0.1	0.714086442875479\\
0.11	0.697792514650011\\
0.12	0.682411964524218\\
0.13	0.667851351632406\\
0.14	0.654032483035044\\
0.15	0.640888999101946\\
0.16	0.62836398040162\\
0.17	0.616408189755405\\
0.18	0.60497863669279\\
0.19	0.594037456234445\\
0.2	0.583551064145412\\
0.21	0.573489500706264\\
0.22	0.563825891653266\\
0.23	0.554536009174625\\
0.24	0.545597904679276\\
0.25	0.536991609512664\\
0.26	0.528698894297002\\
0.27	0.520703058944272\\
0.28	0.512988765765679\\
0.29	0.505541882799538\\
0.3	0.498349361042596\\
0.31	0.491399116250829\\
0.32	0.48467993697205\\
0.33	0.478181397116188\\
0.34	0.471893786330486\\
0.35	0.465808043608739\\
0.36	0.459915704107684\\
0.37	0.454208845773415\\
0.38	0.448680043632942\\
0.39	0.443322333070086\\
0.4	0.438129168952749\\
0.41	0.433094399949728\\
0.42	0.428212238649747\\
0.43	0.4234772358973\\
0.44	0.418884263512407\\
0.45	0.414428491549875\\
0.46	0.410105371997354\\
0.47	0.405910619399287\\
0.48	0.401840197827501\\
0.49	0.397890302911088\\
0.5	0.394057346779817\\
0.51	0.390337943064246\\
0.52	0.386728892281632\\
0.53	0.38322716691302\\
0.54	0.379829898986566\\
0.55	0.376534367406747\\
0.56	0.373337986885974\\
0.57	0.370238255873412\\
0.58	0.367229098705254\\
0.59	0.364293791091734\\
0.6	0.361431220231283\\
};
\addlegendentry{$C_{\text{SC,geom}}(p_d)$}

\addplot [color=black, line width=1.5pt]
  table[row sep=crcr]{%
  0	0.999977980844841\\
0.01	0.941300033523429\\
0.02	0.901602207295377\\
0.03	0.868640294878095\\
0.04	0.839870842613815\\
0.05	0.814105802896766\\
0.06	0.790656046269365\\
0.07	0.769072767581398\\
0.08	0.749041385564589\\
0.09	0.730329827034426\\
0.1	0.712760294103495\\
0.11	0.696192542865622\\
0.12	0.680513366993364\\
0.13	0.665629637698513\\
0.14	0.65146354033389\\
0.15	0.637949183112131\\
0.16	0.6250301410199\\
0.17	0.612657613860639\\
0.18	0.600789035847921\\
0.19	0.589386996897832\\
0.2	0.578418394298117\\
0.21	0.567853762781055\\
0.22	0.557666729756634\\
0.23	0.547833577223651\\
0.24	0.538332877780907\\
0.25	0.529145193661063\\
0.26	0.520252829814418\\
0.27	0.511639616796706\\
0.28	0.503290738301434\\
0.29	0.495192571742188\\
0.3	0.487332563997509\\
0.31	0.479699112725802\\
0.32	0.472281472912097\\
0.33	0.465069668527709\\
0.34	0.458054420169287\\
0.35	0.451227077682366\\
0.36	0.444579565145005\\
0.37	0.43810432737273\\
0.38	0.431794284698287\\
0.39	0.425642795563703\\
0.4	0.419643615659855\\
0.41	0.413790870691\\
0.42	0.40807902474209\\
0.43	0.402502852889921\\
0.44	0.397057421235606\\
0.45	0.391738062447308\\
0.46	0.386540358565513\\
0.47	0.38146012104645\\
0.48	0.376493377518531\\
0.49	0.371636355407316\\
0.5	0.366885469001667\\
0.51	0.362237307357518\\
0.52	0.357688623168088\\
0.53	0.353236322084698\\
0.54	0.348877453751498\\
0.55	0.344609202363056\\
0.56	0.340428880107299\\
0.57	0.336333917743856\\
0.58	0.332321859094719\\
0.59	0.328390354468654\\
0.6	0.324537153737304\\
};
\addlegendentry{$I_{\text{SC,geom}}(p_d) = I_{\text{BSCD,geom}}(0, p_d)$}

\addplot [color=blue, line width=1.5pt]
  table[row sep=crcr]{%
0	0.9194\\
0.01	0.8617\\
0.02	0.8229\\
0.03	0.7903\\
0.04	0.7615\\
0.05	0.737\\
0.06	0.7144\\
0.07	0.6942\\
0.08	0.6766\\
0.09	0.6557\\
0.1	0.6402\\
0.11	0.6234\\
0.12	0.6092\\
0.13	0.5937\\
0.14	0.5795\\
0.15	0.5672\\
0.16	0.5542\\
0.17	0.5437\\
0.18	0.5318\\
0.19	0.5205\\
0.2	0.5098\\
0.21	0.5004\\
0.22	0.4903\\
0.23	0.4816\\
0.24	0.4724\\
0.25	0.4637\\
0.26	0.4554\\
0.27	0.4462\\
0.28	0.4392\\
0.29	0.4296\\
0.3	0.4227\\
0.31	0.414\\
0.32	0.4084\\
0.33	0.4005\\
0.34	0.3942\\
0.35	0.3878\\
0.36	0.3819\\
0.4	0.3552\\
0.41	0.3507\\
0.42	0.3453\\
0.43	0.338\\
0.44	0.3331\\
0.45	0.3284\\
0.46	0.3215\\
0.47	0.3172\\
0.48	0.3125\\
0.49	0.3065\\
0.5	0.3014\\
0.51	0.2964\\
0.52	0.2921\\
0.53	0.288\\
0.54	0.283\\
0.55	0.2801\\
0.56	0.2744\\
0.57	0.2697\\
0.58	0.2656\\
0.59	0.2629\\
0.6	0.257\\
};
\addlegendentry{$I_{\text{BSCD,geom}}(0.01, p_d)$}

\addplot [color=red, line width=1.5pt]
  table[row sep=crcr]{%
0	0.5325\\
0.01	0.4775\\
0.02	0.4472\\
0.03	0.4188\\
0.04	0.3983\\
0.05	0.3784\\
0.06	0.3601\\
0.07	0.3472\\
0.08	0.3337\\
0.09	0.3194\\
0.1	0.3062\\
0.11	0.2956\\
0.12	0.2862\\
0.13	0.2751\\
0.14	0.2673\\
0.15	0.258\\
0.16	0.251\\
0.17	0.2446\\
0.18	0.2372\\
0.19	0.2303\\
0.2	0.2228\\
0.21	0.2176\\
0.22	0.213\\
0.23	0.2073\\
0.24	0.2021\\
0.25	0.198\\
0.26	0.1917\\
0.27	0.187\\
0.28	0.1828\\
0.29	0.1782\\
0.3	0.1742\\
0.31	0.1705\\
0.32	0.1683\\
0.33	0.1637\\
0.34	0.1596\\
0.35	0.1579\\
0.36	0.1541\\
0.4	0.1413\\
0.41	0.1401\\
0.42	0.1367\\
0.43	0.1333\\
0.44	0.1313\\
0.45	0.1304\\
0.46	0.1272\\
0.47	0.1257\\
0.48	0.1237\\
0.49	0.1218\\
0.5	0.1194\\
0.51	0.1168\\
0.52	0.1151\\
0.53	0.1144\\
0.54	0.1115\\
0.55	0.1115\\
0.56	0.1098\\
0.57	0.1062\\
0.58	0.1055\\
0.59	0.1061\\
0.6	0.1036\\
};
\addlegendentry{$I_{\text{BSCD,geom}}(0.1, p_d)$}
\end{axis}
\end{tikzpicture}%
    
    \caption{Monte Carlo estimates of the information rate $I_{\text{BSCD,geom}}(p,p_d)$ of the BSC with error probability $p$, geometric duplications with probability $p_d$, and a $\mathsf{Ber}(1/2)$ source. The information rate $I_{\text{SC,geom}}(p_d)$ is the case when $p=0$, which corresponds to a sticky channel with capacity $C_{\text{SC,geom}}(p_d)$ and is computed numerically \cite{Mitzenmacher2007}.
    }
    \label{fig:bsc_geom}
\end{figure}

\section{Conclusion}
This paper studied the noisy duplication channel and established its channel capacity as the Markov-constrained Shannon capacity in the special case of ergodic Markov sources. It was proven through the AEP for noisy duplication processes, which was related to the AEP for HSMPs. The motivation for these results was to provide a theoretical underpinning to the numerical results in \cite{McBain2022, McBain2024}.
A significant open problem that remains is the construction of codes that can achieve the Markov-constrained capacity in practice, which is still in its early stages \cite{Vidal2023,Vidal2023_b,McBain2023,Vidal2024}. Further progress in this area could greatly advance DNA storage systems based on nanopore sequencing.

 \appendices

 \section{Preliminary lemmas}

\begin{lemma_appendix}
$H(Y,A) - H(Y) \leq 2H(A)$.
\label{ent_diff_ineq}
\end{lemma_appendix}
\begin{proof}
    We start from the inequality $|H(Y) - H(Y|A)| \leq H(A)$ \cite[~Lemma 4A.1]{Gallager1968}. Then $H(Y,A) - H(Y) = H(Y|A) + H(A) - H(Y)
    = |H(Y|A) + H(A) - H(Y)| 
    \leq |H(Y|A) - H(Y)| + H(A) \leq 2H(A)$.
\end{proof}

\begin{lemma_appendix}
$\lim_{d \rightarrow \infty} \frac{1}{d} H(T_d) = 0$.
\label{ent_jump_time}
\end{lemma_appendix}
\begin{proof}
    By Hoeffding's inequality, we have the concentration inequality $\mathbb{P}(|T_d - \mathbb{E}[T_d]| > d) \leq 2 \exp(-C d)$ for some constant $C > 0$. Then the proof follows using the same argument as in Lemma \ref{lem:rand_indexed_ent_rate}. Alternatively, we could use the Gaussian approximation.\end{proof}

\section{Proof of Theorem \ref{output_AEP}}\label{appendixB}

\textbf{Split problem into three parts.} We want to show that for any $\epsilon > 0$ and $\delta > 0$ we have $\mathbb{P}(|g_m - H| > \epsilon) 
< \delta$ for sufficiently large $m$. Let $\Delta_1 = \{|g_{m,d} - g_m|
    > \epsilon/3\}$, $\Delta_2 = \{|g_{m,d} - H_{\infty,d}|
    > \epsilon/3\}$, and $\Delta_3 = \{|H_{\infty,d} - H| > \epsilon/3\}$. Then we have the bound
\begin{align}
    &\mathbb{P}(|g_m - H| > \epsilon) \\
    &\leq \mathbb{P}(|g_{m,d}-g_m| + |g_{m,d} - H_{\infty,d}| + |H_{\infty,d} - H| > \epsilon) \\
    &\leq \underbrace{\mathbb{P}(\Delta_1)}_{\text{Term 1}} + \underbrace{\mathbb{P}(\Delta_2)}_{\text{Term 2}} + \underbrace{\mathbb{P}(\Delta_3)}_{\text{Term 3}}
\end{align}
 and we want to show that each term is bounded by $\delta/3$, where $d$ is fixed and does not affect our choice of $m$. We will prove this in three parts and then combine the results.

\vspace{2mm}
\textbf{Convergence of entropy rates with markers.}  Observe that $g_{m,d}$ is decreasing towards $g_{m}$ for increasing $d$ (with equality at $d=m$), then $g_{m,d} - g_{m} \geq 0$ and $\mathbb{E}[g_{m,d} - g_{m}] \geq 0$. From Lemma A\ref{ent_diff_ineq}, for a fixed $d$ we have the bound $\mathbb{E}[g_{m,d} - g_{m}] = H_{m,d} - H_m \leq \frac{2}{m}H(\mathcal{M}_{m,d}) = \frac{2}{d}H(T_d)$ uniformly in $m$. Then take the limit in $m$ to get $H_{\infty,d} - H \leq \frac{2}{d}H(T_d)$. By Lemma A\ref{ent_jump_time}, there exists a $D$ such that $H_{m,d} - H \leq \epsilon/3$ and then $\Delta_3$ occurs with probability $0$ for all $d \geq D$. This proves almost sure convergence for Term 3.

\vspace{2mm}
\textbf{Convergence of sample entropies with markers.} For any $\delta' > 0$ there exists a $D'$ such that such that $H_{m,d} - H_m \leq \delta'$ and $\delta'/\epsilon \leq \delta/3$ for all $d\geq D'$. Recalling that $g_{m,d} - g_{m} \geq 0$, we can apply Markov's inequality to get
\begin{align}
    \mathbb{P}(\Delta_1) &\leq 3\mathbb{E}[g_{m,d} - g_{m}]/\epsilon\\
    &\leq 3 (H_{m,d} - H_m)/\epsilon \\
    &\leq \delta' /\epsilon \leq \delta/3
\end{align}
for all $d \geq D'$. This proves convergence in probability for Term 1, which will now be linked with the convergence result of Term 3 by showing the convergence of Term 2. 
\vspace{2mm}

\textbf{Shannon-McMillan-Breiman theorem with markers.} Consider the $d$-step MP $\{V_n\}$ where $V_n = S_{(n-1)d + 1}^{n d}$. Since it embeds the ergodic MP $\{S_{\ell}\}$, which is a one-to-one mapping between probability spaces with identical probability measure, it is also ergodic. Then we have the HMP $\{W_n\}$ where $W_n = Y_{T_{(n-1)d}+1}^{T_{nd}}$. Since the $d$-step MP is ergodic, then the resulting HMP is ergodic \cite{Leroux1992}. Apply SMB to get $g_{nd,d}\rightarrow H_{\infty,d}$ almost surely as $n\rightarrow\infty$ for a fixed $d$ (hence $m\rightarrow\infty$). Then there exists a sequence $\{N(d)\}$ such that $|g_{n d,d} - H_{\infty,d}| \leq \epsilon/6$ for all $n \geq \overline{N}(d) = \max \{N(n_1): n_1 \leq d \}$ with probability $1$. We force monotonicity in $\{\overline{N}(d)\}$ in order to later specify a lower bound on $m$ in terms of a lower bound on $d$.

By applying the chain rule, observe that $g_{m,d} = g_{nd, d} + g'_{i}$ for a sample entropy $g'_i \geq 0$ with mean bounded as $\mathbb{E}[g'_{i}] \leq \frac{1}{m} H(Y_{1}^{T_{i}}) \leq C/m$ for some constant $C>0$. By Markov's inequality, we have $\mathbb{P}(g'_{i} \geq \epsilon/6) \leq 6C/m\epsilon$ and hence there exists an $M$ such that $6C/m\epsilon \leq \delta/3$ for all $m\geq M$. Then
\begin{align}
    \mathbb{P}(\Delta_2) &\leq \mathbb{P}(|g_{nd,d} - H_{\infty, d}| + g'_i > \epsilon/3)\\
    &\leq \mathbb{P}(|g_{nd,d} - H_{\infty, d}| > \epsilon/6) + \mathbb{P}(g'_i > \epsilon/6)\\
    &\leq 6C/m\epsilon \leq \delta/3
\end{align}
for all $m \geq \max\{\overline{N}(d) d, M\}$ and for any fixed $d$, where we use $m \geq nd$ to specify a lower bound on $m$ in terms of $\overline{N}(d)$. This proves convergence in probability for Term 2.

\vspace{2mm}

Combining all three terms, we have $\mathbb{P}(|g_m - H| > \epsilon) 
 \leq 2\delta/3$ for all $m \geq \max\{\overline{N}(d) d, M\}$ and for all $d \geq \overline{D} = \max\{D,D'\}$. Hence, for any $\epsilon > 0$ and $\delta >0$ we choose $d=\overline{D}$ so that $\mathbb{P}(|g_m - H| > \epsilon) 
 < \delta$ for all $m \geq \max\{\overline{N}(\overline{D}) \overline{D}, M\}$.

\balance
\bibliographystyle{plain}
\bibliography{refs}

\end{document}